\newif\ifmai\maifalse
\newif\iffull\fulltrue

\iffull
\documentclass{article}
\else
\documentclass{llncs}
\fi

\usepackage{graphicx}
\usepackage{amssymb,amsmath}
\usepackage{array,multirow}

\newcommand{\defi}{\stackrel{\triangle}{=}}

\newcommand{\qif}{{\sc  qif}}

\newcommand{\es}{\emptyset}

\iffull
\newtheorem{definition}{Definition}[section]
\newtheorem{proposition}[definition]{Proposition}
\newtheorem{lemma}[definition]{Lemma}
\newtheorem{corollary}[definition]{Corollary}
\newtheorem{theorem}[definition]{Theorem}

\newcommand{\qed}{\hfill\mbox{\ \vrule height4pt width3pt depth2pt}}
\newenvironment{proof}%
{\begin{trivlist}%
\item[]{\bf Proof:}}%
{\end{trivlist}}

\fi

 \newcommand{\F}{\mathcal{F}}
 \newcommand{\Code}{\mathcal{C}}

\newcommand{\var}{\mathrm{var}}

\title{Approximate Model Counting, Sparse XOR Constraints  and Minimum Distance}         

\iffull
\author{Michele Boreale$^1$ \quad Daniele Gorla$^2$
\\$^1$Univ. Firenze, Dip. Statistica, Informatica, Applicazioni (DiSIA) 
\\
$^2$Universit\`a di Roma ``La Sapienza",  Dep. of Computer Science            
}
\else
\author{Michele Boreale$^1$ \and Daniele Gorla$^2$}
\institute{
$^1$Universit\`a  di Firenze, Dip. Statistica, Informatica, Applicazioni (DiSIA) 
$^2$Universit\`a di Roma ``La Sapienza",  Dep. of Computer Science            
}
\fi

\begin{document}

\maketitle

\begin{abstract}
The problem of counting the number of models of a given Boolean
formula has numerous applications, including computing the leakage
of deterministic programs in Quantitative Information Flow. Model
counting is a hard,   \#P-complete problem. For this reason, many
approximate counters have been developed in the last decade,
offering formal guarantees of confidence and accuracy. A popular
approach is based on the idea of using random XOR constraints  to,
roughly, successively halving the solution set until no model is
left: this is checked by invocations to a SAT solver. The
effectiveness of this procedure hinges on the ability of the SAT
solver to deal with XOR constraints, which in turn crucially depends
on the length of such constraints. We study   to what extent one can
employ sparse, hence short, constraints, keeping   guarantees of
correctness. We show that the resulting bounds are closely related
to the geometry of the set of models,   in particular to the minimum
Hamming distance between models. We evaluate our theoretical results
on a few concrete formulae. Based on our findings, we finally
discuss
  possible directions for
  improvements of the current state of the art in approximate model counting.
\end{abstract}


\section{Introduction}\label{sec:intro}
\#SAT (aka {\em model-counting}) is the problem of counting the
number of satisfying assignments of a given Boolean formula and is a
\#P-complete problem. Indeed, every NP Turing machine can be encoded
as a formula
  whose satisfying assignments correspond to the accepting paths of  the machine  \cite{Val79}.
Thus, model-counting is harder than satisfiability: \#SAT is indeed
intractable in cases for which SAT is tractable (e.g., sets of Horn
clauses or sets of 2-literal clauses) \cite{VV86}. Still, there are
cases in which model-counting is tractable (e.g., OBDDs and
d-DNNFs). For a very good overview of the problem and of some
approaches to it see \cite{GSS09}.

Our interest in model counting originates from its applications in
the field of Quantitative Information Flow (\qif)
\cite{Catuscia,Smith}. Indeed, a   basic result in \qif\ is that the
maximum min-entropy leakage of a deterministic program is $\log_2
k$, with $k$ the number of distinct outputs the program can return
\cite{Smith}, varying the input. If the program is modeled as a
Boolean formula, then computing its leakage reduces
 to \#SAT, specifically to computing the number of models of the formula  obtained by
 existentially projecting out the non-output variables; see \cite{Legay18,KWW16}. 

Over the years, several \emph{exact} counting algorithms have been
put forward and implemented, such as, among others,
\cite{GKNS07,KWW16,MMBH12,Thu06}, with applications to \qif\
\cite{KMM13}. The problem with exact counters is that, although
performing reasonably  well when certain parameters of the formula
-- size, number of variables, number of clauses -- are relatively
small, they rapidly go out of memory as these parameters grow.

For this reason, 
\emph{approximate}  counters have more and more been considered.
Indeed, in many applications,  the exact count of models is not
 required: it may suffice to provide an estimate, as long as
the method is quick and it is   equipped with a formal guarantee of
correctness. This is typically the case in \qif, where knowing the
exact count within a factor of $\eta$ is sufficient to estimate
leakage within $\log_2\eta$ bits. For probabilistic counters,
correctness is usually expressed in terms of two parameters:  {\em
accuracy} -- the desired maximum difference between the reported and
the true count; and {\em confidence} -- the probability that the
reported result is actually within the specified accuracy.
%
%

We set ourselves in the  line of research pioneered by \cite{VV86}
and followed,
 e.g., by \cite{AT17,Gomes,approxmc,KWW16}.
 The basic idea of a probabilistic model counting algorithm
  is the following (Section \ref{sec:two}):
 given a formula $\phi$
in the Boolean variables $y_1,\ldots,y_m$, one  chooses at random
$\langle a_0,\ldots, a_m\rangle \in \{0,1\}^{m+1}$. The resulting
  {\em XOR constraint}
$a_0  = a_1y_1\oplus \cdots \oplus a_m y_m$
splits evenly the set of models of $\phi$ into two parts: those
satisfying the constraint and those not satisfying it. If one
independently generates $s$ such constraints $c_1,\ldots,c_s$, the
formula $\phi'
\defi \phi\wedge c_1\wedge\cdots \wedge c_s$ has an expected $\frac N
{2^s}$ models, where $N$ is the number of models of $\phi$ (i.e.,
the number we aim at estimating). If $\phi'$ is still satisfiable,
then with high probability $N\geq 2^s$, otherwise $N<2^s$. By
repeating this process, one can arrive at a good estimate of $N$.
This procedure can be implemented by  relying on  any SAT-solver
capable of dealing with XOR constraints,  e.g  CryptoMiniSat
\cite{CMS}; or even converting the XOR constraints into CNF before
feeding $\phi'$ to the SAT solver. In any case,   the branching
factor associated with searching for models of $\phi'$ quickly
explodes as the \emph{length} (number of variables)  of  the XOR
constraints grows. The random
  generation outlined above will lead to an expected length
of $\frac m 2$ for each constraint, making  the procedure not easily
scalable as $m$ grows.

In the present paper,  we study   under what conditions one can
employ   sparse, hence shorter, constraints, keeping the same
guarantees of correctness. We generalize the results of
\cite{EGSS14} to  arrive at an improved understanding of how
sparsity is related to minimum distance between models, and how this
affects  the   counting procedure. Based on these results, we also
suggest a possible direction for a new counting methodology based on
the use of Low Density Parity Check (LDPC) codes \cite{Gall62,LDPC};
however, we leave a through experimentation with this new
methodology for future work.

The main point is to generate the coefficients $a_1,\ldots,a_m$
according to a probability value $\lambda \in \left(0,\frac 1 2\right]$, rather
than uniformly. This way, the constraints will have an average
length of $\lambda m\ \leq \frac m 2$   each.
Basically,  the correctness guarantees of the algorithm depend, via
 the Chebyshev inequality, on  keeping  the  variance
   of the number of models of $\phi'$, the formula
obtained by joining the constraints, below a certain threshold. A
value of the density $\lambda$ that achieves this is said to be
\emph{feasible} for the formula. In our main result (Section
\ref{sec:three}), we provide a bound    on the variance that   also
depends on the minimum Hamming distance $d$ between the formula's
models: a larger $d$ yields a smaller variance, hence smaller
feasible $\lambda$'s. Our bound essentially coincides with that of
\cite{EGSS14} for $d=1$.
Therefore, in principle, a lower bound on the minimum distance can
be used to obtain tighter bounds on $\lambda$, making the XOR
constraints shorter and the counting procedure  pragmatically more
efficient.

We will show this phenomenon at work (Section
\ref{sec:four}) on some formulae where the value of $d$ is known by
construction, comparing our results with the state of the art model
counter ApproxMC3 \cite{SM19}. These considerations also
suggest that, if no information on $d$ is available, one can encode
the formula's models using an error correcting code with a known
minimum distance. We will briefly discuss the use of LDPC codes to
this purpose, although  at the moment we have no experimental
results available in this respect. A comparison with recent related
work concludes the paper (Section \ref{sec:five}). 
%


\section{A general counting algorithm}\label{sec:two}
In what follows, we let $\phi(y)$, or just $\phi$,   denote a
generic boolean formula 
with boolean variables $y=(y_1,...,y_m)$ and $m\geq 1$.

\subsection{A general scheme}
According to a well-known general scheme \cite{Gomes},  the building
block of a statistical counting procedure is a probabilistic
decision algorithm: with high probability, this algorithm correctly
decides whether the cardinality of the set is, or is not, below a
given threshold $2^s$, within some tolerance factors, given by the
slack parameters $\alpha$ and $\beta$ below.

\begin{definition}[\#SAT decision algorithm]\label{def:decalg}  Let $0\leq \delta < frac 1 2$ (error probability), $\alpha > 1$ and $\beta > 1$ (two slack parameters) be three reals. An \emph{$(\alpha,\beta,\delta)$-decision algorithm} (for \#SAT) is a  probabilistic algorithm $A(\cdot,\cdot)$, taking a pair of an integer $s$ and a boolean formula $\phi$, and returning either $1$ (meaning `$\#\phi \geq 2^{s-\alpha}$')   or $0$ (meaning `$\#\phi \leq 2^{s+\beta}$') and
such that for each integer $s\geq 0$ and formula $\phi$:
\begin{enumerate}
\item $\#\phi > 2^{s+\beta}$ implies $\Pr\left(A(s,\phi)=0\right)\leq \delta$;
\item $\#\phi < 2^{s-\alpha}$ implies $\Pr\left(A(s,\phi)=1\right)\leq \delta$.
\end{enumerate}
\end{definition}

The use of two different slack parameters in the definition above is justified by the need of stating
formal guarantees about the outcome of the algorithm, while keeping the precision of the algorithm
as high as possible.

As usual, we can boost the confidence in the reported answer, and get an arbitrarily small  error probability, by running   $A(s,\phi)$  several times independently. In particular, consider the algorithm $RA_t(s,\phi)$ obtained by running $A(s,\phi)$ an odd $t\geq 1$ number times independently, and then reporting the majority answer. Call $Err$ the event that $RA_t$ reports a wrong answer; then,
\begin{eqnarray}
\nonumber
\Pr(Err)
& = & \Pr(\text{at least $\left\lceil \frac t 2 \right\rceil$ runs of $A(s,\phi)$ report the wrong answer})
\\
\nonumber
& = & \sum_{k=\left\lceil \frac t 2 \right\rceil}^t \Pr(\text{exactly $k$ runs of $A(s,\phi)$ report the wrong answer})
\\
& = & \sum_{k=\left\lceil \frac t 2 \right\rceil}^t \binom t k p^k (1-p)^{t-k}
\label{eq:exprErr}
\end{eqnarray}
where
\begin{eqnarray}
\nonumber
p & \defi & \Pr(A(s,\phi) \text{ reports the wrong answer})
\\
\nonumber
&=&
\left\{
\begin{array}{ll}
\Pr(A(s,\phi)= 0) & \qquad \mbox{if } \#\phi > 2^{s+\beta}\\
\Pr(A(s,\phi)= 1) & \qquad\mbox{if } \#\phi < 2^{s-\alpha}
\end{array}
\right.
\\
& \leq& \delta
\label{eq:prSingleErr}
\end{eqnarray}
Now, replacing \eqref{eq:prSingleErr} in \eqref{eq:exprErr}, we obtain
\begin{eqnarray}
\Pr(Err)
& \leq & \sum_{k=\left\lceil \frac t 2 \right\rceil}^t \binom t k \delta^k (1-\delta)^{t-k}
\label{eq:prErr}
\end{eqnarray}
Let us call $\Delta(t,\delta)$ the right hand side of \eqref{eq:prErr};
then, $RA_t$ is an $(\alpha,\beta,\Delta(t,\delta))$-decision algorithm
whenever $A$ is an $(\alpha,\beta,\delta)$-decision algorithm.


We now show that any $(\alpha,\beta,\delta)$-decision algorithm $A$
for \#SAT can be used as a  building block for a counting algorithm,
$C_A(\phi)$, that determines an interval  $[\ell,u]$ such that
$\lfloor 2^\ell \rfloor \leq \#\phi \leq \lceil 2^u \rceil$ with
high probability.
%
%
%
Informally, starting with an initial interval $[-1, m]$, the
algorithm $C_A$ performs a binary search,
  using $A$ to decide which half of the current interval $\log_2(\#\phi)$ lies in. The search stops when the current
  interval  cannot be further narrowed, taking into account the slack parameters $\alpha$ and $\beta$, or when a certain
  predefined number of iterations is reached.
Formally,  let $I_0\defi [-1, m]$. Assume $k>0$ and $I_k=[l_k,u_k]$,
then: 
\begin{itemize}
\item[(a)] if $u_k - l_k \leq 2\max(\alpha,\beta)+1$ or $k = \lceil \log_2(m)\rceil$, then return $I_k$;
\item[(b)] otherwise, let $s=\mathit{round}\!\left(\frac{u_k+l_k}2\right)$;  if $A(s,\phi)=0$
then $I_{k+1}\defi [l_k,s+\beta]$   otherwise  $I_{k+1}\defi [s-\alpha,u_k]$.
\end{itemize}

\begin{theorem}\label{th:CA}
Let $A$ be a $(\alpha,\beta,\delta)$-decision algorithm. Then:
\begin{enumerate}
\item $C_A(\phi)$ terminates in   $k\leq \lceil \log_2 m \rceil$ iterations
  returning an interval  $I_k=[l,u]$ such that  $u-l\leq 2\max(\alpha,\beta)+2$;
\item The probability that
$\#\phi\notin [\lfloor 2^l \rfloor, \lceil 2^u\rceil]$ is at most $ \lceil \log_2 m \rceil \delta$.
\end{enumerate}
\end{theorem}
\iffull
\begin{proof}
If the algorithm terminates because $u_k - l_k \leq 2\max(\alpha,\beta)+1$, the first claim is trivial.
Otherwise, by construction of the algorithm, we have that
\begin{eqnarray}
|I_0| & = & m+1
\label{eq:I0}
\end{eqnarray}
Furthermore, by passing from $I_{k-1}$ to $I_{k}$, we have that
$s_k =round\!\left( \frac{u_{k-1}+l_{k-1}}2\right)$ and
$$
\begin{array}{rl}
|I_{k}| \leq &
\left\{
\begin{array}{ll}
\frac {u_{k-1}-l_{k-1}} 2 + \alpha + \frac 1 2 & \mbox{ if } A(s_k,m) = 1 \mbox{ and } s_k = \left\lfloor \frac {u_{k-1}-l_{k-1}} 2 \right\rfloor
\vspace*{.2cm}
\\
\frac {u_{k-1}-l_{k-1}} 2 + \beta & \mbox{ if } A(s_k,m) = 0  \mbox{ and } s_k = \left\lfloor \frac {u_{k-1}-l_{k-1}} 2 \right\rfloor
\vspace*{.2cm}
\\
\frac {u_{k-1}-l_{k-1}} 2 + \alpha & \mbox{ if } A(s_k,m) = 1 \mbox{ and } s_k = \left\lceil \frac {u_{k-1}-l_{k-1}} 2 \right\rceil
\vspace*{.2cm}
\\
\frac {u_{k-1}-l_{k-1}} 2 + \beta + \frac 1 2 & \mbox{ if } A(s_k,m) = 0  \mbox{ and } s_k = \left\lceil \frac {u_{k-1}-l_{k-1}} 2 \right\rceil
\end{array}
\right.
\end{array}
$$
since, by definition of $round$, $\frac{u_{k-1}+l_{k-1}}2 - \frac 1 2 < s_k \leq \frac{u_{k-1}+l_{k-1}}2$, if $s_k = \left\lfloor \frac{u_{k-1}+l_{k-1}}2 \right\rfloor$, and $\frac{u_{k-1}+l_{k-1}}2 \leq s_k < \frac{u_{k-1}+l_{k-1}}2  + \frac 1 2$, if $s_k = \left\lceil \frac{u_{k-1}+l_{k-1}}2 \right\rceil$.
Thus, by letting $M = \max(\alpha,\beta) + \frac 1 2$, we have:
\begin{eqnarray}
|I_k| & \leq & \frac{|I_{k-1}|} 2 + M
\label{eq:Ik}
\end{eqnarray}
If we now unfold \eqref{eq:Ik} and use \eqref{eq:I0}, we obtain that
\begin{eqnarray}
|I_k| & \leq & \frac 1 2 |I_{k-1}| + M
\nonumber\\
& \leq & \frac 1 2 \left(  \frac 1 2 |I_{k-2}| + M \right) + M
= \frac 1 4 |I_{k-2}| +\frac 3 2 M
\nonumber\\
\nonumber\\
& \ldots &
\nonumber\\
& \leq & \frac 1 {2^k} |I_0| + M\sum_{i=0}^{k-1} \frac 1 {2^i}
= \frac {m+1} {2^k} + M\left( 2 - \frac 1 {2^{k-1}} \right)
\label{eq:exprIk}
\end{eqnarray}
By now considering $k = \lceil\log_2 m\rceil$ in \eqref{eq:exprIk}, we obtain that
\begin{eqnarray*}
|I_{\lceil\log_2 m\rceil}|
& \leq & \frac {m+1} {2^{\lceil\log_2 m\rceil}} + M\left( 2 - \frac 1 {2^{\lceil\log_2 m\rceil-1}} \right)
\\
& \leq & \frac {m+1} {m} + M\left( 2 - \frac 1 {m} \right)
\\
& \leq & 2M+1 = 2\max(\alpha,\beta) + 2
\end{eqnarray*}
being $\log_2 m \leq \lceil\log_2 m\rceil \leq \log_2 m +1$
and $M > 1$.

The error probability is the probability that either
$\#\phi< \lfloor 2^l \rfloor$ or $\#\phi > 2^u$.
This is the probability that one of the $\lceil\log_2 m\rceil$ calls to $A$ has returned a wrong answer,
that can be calculated as follows:
$$
\sum_{k=1}^{\lceil\log_2 m\rceil} \Pr(OK_1)\ldots\Pr(OK_{k-1})\Pr(ERR_k)
=
\sum_{k=1}^{\lceil\log_2 m\rceil} (1-\delta)^{k-1}\delta
\leq
\lceil\log_2 m\rceil \delta
$$
where $\Pr(OK_i)$/$\Pr(ERR_i)$ denotes the probability that the $i$-th iteration
returned a correct/wrong answer, respectively.
\else
\begin{proof}[Sketch]
If the algorithm terminates because $u_k - l_k \leq 2\max(\alpha,\beta)+1$, the first claim is trivial.
Otherwise, by construction of the algorithm, we have that
$|I_0| = m+1$. Furthermore, by passing from $I_{k-1}$ to $I_{k}$, we have that
$s_k =round\!\left( \frac{u_{k-1}+l_{k-1}}2\right)$ and
$$
\begin{array}{rl}
|I_{k}| \leq &
\left\{
\begin{array}{ll}
\frac {u_{k-1}-l_{k-1}} 2 + \alpha + \frac 1 2 & \mbox{ if } A(s_k,m) = 1 \mbox{ and } s_k = \left\lfloor \frac {u_{k-1}-l_{k-1}} 2 \right\rfloor
\vspace*{.2cm}
\\
\frac {u_{k-1}-l_{k-1}} 2 + \beta & \mbox{ if } A(s_k,m) = 0  \mbox{ and } s_k = \left\lfloor \frac {u_{k-1}-l_{k-1}} 2 \right\rfloor
\vspace*{.2cm}
\\
\frac {u_{k-1}-l_{k-1}} 2 + \alpha & \mbox{ if } A(s_k,m) = 1 \mbox{ and } s_k = \left\lceil \frac {u_{k-1}-l_{k-1}} 2 \right\rceil
\vspace*{.2cm}
\\
\frac {u_{k-1}-l_{k-1}} 2 + \beta + \frac 1 2 & \mbox{ if } A(s_k,m) = 0  \mbox{ and } s_k = \left\lceil \frac {u_{k-1}-l_{k-1}} 2 \right\rceil
\end{array}
\right.
\end{array}
$$
Thus, by letting $M = \max(\alpha,\beta) + \frac 1 2$, we have that $|I_k| \leq \frac{|I_{k-1}|} 2 + M$;
this suffices to obtain $|I_{\lceil\log_2 m\rceil}| \leq 2\max(\alpha,\beta) + 2$.

The error probability is the probability that either
$\#\phi< \lfloor 2^l \rfloor$ or $\#\phi > 2^u$;
this is the probability that one of the $\lceil\log_2 m\rceil$ calls to $A$ has returned a wrong answer.
\fi
\qed
\end{proof}

In all the experiments we have run, the algorithm has always
returned an interval of width at most $2\max(\alpha,\beta)+1$,
sometimes in   less than $\lceil\log_2(m)\rceil$ iterations: this
consideration pragmatically justifies the exit condition we used in
the algorithm. 
We leave for future work a more elaborated analysis of the algorithm
to formally establish   the $2\max(\alpha,\beta)+1$ bound.

\subsection{XOR-based decision algorithms}\label{subs:xorbased}

Recall that a XOR constraint $c$ on the variables $y_1,...,y_m$ is an equality  of the form
\begin{eqnarray} \nonumber
a_0 & = & a_1y_1\oplus \cdots \oplus a_m y_m
\end{eqnarray}
where $a_i\in \mathbb{F}_2$ for $i=0,...,m$ (here $\mathbb{F}_2=\{0,1\}$ is the two elements field.) Hence $c$ can be identified with a $(m+1)$-tuple in $\mathbb{F}^{m+1}_2$.
Assume that a probability distribution on $\mathbb{F}^{m+1}_2$ is fixed. A simple proposal for a decision algorithm $A(s,\phi)$ is as follows:
\begin{enumerate}
\item generate $s$ XOR constraints $c_1,\ldots,c_s$ independently,   according to the fixed probability distribution;
\item if $\phi\wedge c_1\wedge\cdots \wedge c_s$ is unsatisfiable then return 0, 
else return 1. 
\end{enumerate}
Indeed \cite{Gomes}, every XOR constraint  splits the set of boolean
assignments in two parts, according to whether the assignment
satisfies the constraint or not. Thus, if $\phi$ has less than $2^s$
models (and so less than $2^{s+\beta}$), the formula $\phi\wedge
c_1\wedge\cdots \wedge c_s$ is likely to be unsatisfiable.

In step 2 above, any off-the-shelf SAT  solver can be employed: one
appealing possibility is using CryptoMiniSat \cite{CMS}, which
offers support for specifying XOR constraints (see e.g.
\cite{Legay18,approxmc,KWW16}). Similarly to \cite{Gomes}, it can be
proved that this algorithm yields  indeed an $
(\alpha,\beta,\delta)$-decision algorithm, for a suitable
$\delta<\frac 12$, if the constraints $c_i$ at step 1 are chosen
uniformly at random. This however will generate `long' constraints,
with an average of $\frac m 2$ variables each, which a SAT solver will not
be able to manage as $m$ grows.

\section{Counting with sparse XORs}\label{sec:three}

We want to explore alternative ways of generating     constraints,
which will   make it possible to work with short (`sparse') XOR
constraints, while keeping the same guarantees of correctness. In
what follows, we assume a probability distribution over the
constraints, where each coefficient $a_i$, for $i>0$, is chosen independently with
probability $\lambda$, while $a_0$ is chosen uniformly (and independently from all the other $a_i$'s). 
In other words, we assume that the
probability distribution over $\mathbb{F}^{m+1}_2$ is of the
following form, for $\lambda\in \left(0,\frac12\right]$:
\begin{eqnarray}\label{eq:distr}
\Pr(a_0,a_1,...,a_m) & \defi & p(a_0)  p'(a_1)\cdots p'(a_m)\\
\text{ where: }&& p(1)=p(0)=\frac 1 2 \qquad p'(1)=\lambda\qquad p'(0)=1-\lambda\nonumber
\end{eqnarray}
The
expected number of variables appearing in a constraint chosen
according to this distribution will therefore be $m\lambda$. Let us still
call $A$ the algorithm presented in Section \ref{subs:xorbased},
with this strategy in choosing the constraints. We want to establish
conditions on $\lambda$ under which $A$ can be proved to be a
decision algorithm.

Throughout this section, we fix a boolean formula $\phi(y_1,...,y_m)$ and
   $s\geq 1$ 
XOR constraints. Let
$$\chi_s\defi \phi\wedge c_1\wedge\cdots \wedge c_s$$
where the $c_i$ are chosen independently according to
\eqref{eq:distr}.   For any assignment (model)  $\sigma$ from
variables $y_1,...,y_m$ to $\mathbb{F}_2$, let us denote by
$Y_\sigma$ the Bernoulli r.v. which is $1$ iff $\sigma$ satisfies
$\chi_s$.
 
We now list the steps needed for proving that $A$ is a decision algorithm.
This latter result is obtained by Proposition \ref{prop:variance}(2) (that derives from Lemma \ref{lemma:basic}(1))
and by combining Proposition \ref{prop:variance}(1),
Lemma \ref{cor:variance}(3) (that derives from Lemma \ref{lemma:basic}(2)) and Lemma \ref{lemma:S}  
in Theorem \ref{th:mainBound} later on.
In what follows, we shall let $\rho\defi 1-2\lambda$.

\begin{lemma}\label{lemma:basic}\
\begin{enumerate}
\item $\Pr(Y_\sigma =1)= E[Y_\sigma]= 2^{-s}$.
\item Let $\sigma,\sigma'$ be any two assignments and $d$  be
their Hamming distance in $\mathbb{F}_2^m$ (i.e., the size of their
symmetric difference  seen as subsets of $\{1,...,m\}$). Then
$\Pr(Y_\sigma =1,Y_{\sigma'} =1)=E[Y_\sigma\cdot
Y_{\sigma'}]=\left(\frac{1+\rho^d}4\right)^s$ (where we let $\rho^d\defi 1$ whenever
$\rho=d=0$.)
\end{enumerate}
\end{lemma}
\iffull
\begin{proof}
Concerning the first item, note that the probability that $\sigma$ satisfies any constraint chosen according to \eqref{eq:distr} is $\frac 12$, because the parity bit $a_0$ is chosen uniformly at random; consequently, the probability that $\sigma$ satisfies the $s$ constraints $c_1,\ldots,c_s$  chosen independently is $\frac 1 {2^s}$.

Let us examine the second item. We first consider the case of $s=1$ constraint. Let $A$ and $B$ be the sets of variables that are assigned the value 1 in $\sigma$ and $\sigma'$, respectively.   We let   $U=A\setminus B$, $V=B\setminus A$ and $I=A\cap B$. Note that by definition $d=|U\cup V|=|U|+|V|$.  Let $C$ be the set of  variables appearing in the constraint $c_1$.
Assume first that the constraint's parity bit $a_0$ is $0$, which happens with probability $\frac 12$. Given this event, both $\sigma$ and $\sigma'$ satisfy the constraint if and only if both $|C\cap A|$ and $|C\cap B|$ are even. This is in turn equivalent to   (all of $|C\cap U|, |C\cap I|$, $|C\cap V|$ are even)  or (all of  $|C\cap U|, |C\cap I|$, $|C\cap V|$ are odd). Moreover these two events are clearly disjoint. Since the three sets of variables $U,I,V$ are pairwise disjoint, and since according to \eqref{eq:distr} the constraint's variables are chosen independently, we can compute as follows. Here, $U_e$ and $U_o$ abbreviate $\Pr( |C\cap U|\text{ is even})$ and $\Pr( |C\cap U|\text{ is odd})$, respectively; similarly for $I_e,I_o$ and $V_e,V_o$.
\begin{eqnarray}
&&\Pr(Y_\sigma=1,Y_{\sigma'}=1|a_0=0)\ = \nonumber\\
&&\qquad =\ \Pr(\text{ $|C\cap U|, |C\cap I|$, $|C\cap V|$ are   even})\nonumber\\
&&  \qquad\quad \ + \ \Pr(\text{$|C\cap U|, |C\cap I|$, $|C\cap V|$ are   odd})\nonumber\\
&&\qquad =\  U_e\cdot I_e\cdot V_e+U_o\cdot I_o\cdot V_o\label{eq:one}
\end{eqnarray}
Reasoning similarly, we obtain
\begin{eqnarray}
&& \Pr(Y_\sigma=1,Y_{\sigma'}=1|a_0=1) \nonumber\\
&& \qquad =\ \Pr(\text{$|C\cap I|$ is even, $|C\cap U|, |C\cap I|$ are odd})\nonumber\\
&&  \qquad\quad +\ \Pr(\text{$|C\cap I|$ is odd, $|C\cap U|, |C\cap I|$ are even})\nonumber\\[5pt]
&& \qquad =\ U_o\cdot I_e\cdot V_o + U_e\cdot I_o\cdot V_e  \nonumber\\[5pt]
&& \qquad =\ U_o\cdot (1-I_o)\cdot V_o + U_e\cdot(1- I_e)\cdot V_e  \nonumber  \\[5pt]
&& \qquad =\ U_o\cdot V_o + U_e\cdot V_e-  U_o\cdot I_o\cdot V_o - U_e\cdot I_e\cdot V_e   \label{eq:two}
\end{eqnarray}
By using \eqref{eq:one} and \eqref{eq:two}, we can compute
\begin{eqnarray}
\Pr(Y_\sigma=1,Y_{\sigma'}=1 ) & = & \frac 1 2\Pr(Y_\sigma=1,Y_{\sigma'}=1|a_0=0)  \nonumber\\
&& +\ \frac 1 2 \Pr(Y_\sigma=1,Y_{\sigma'}=1|a_0=1)\nonumber\\[5pt]
                                                             & = & \frac 1 2 \left(U_o\cdot V_o + U_e\cdot V_e \right)\nonumber\\[5pt]
                                                             & = & \frac 1 2 \left((1-U_e)\cdot (1-V_e) + U_e\cdot V_e\right)\nonumber\\[5pt]
                                                             & = & \frac 1 2 \left( 1- V_e - U_e + 2 U_e\cdot V_e \right)\label{eq:three}
\end{eqnarray}
%
Now, it is a standard result (see e.g. \cite{mathexch}) that, given $t$ independent trials of a Bernoulli variable of parameter $\lambda\leq \frac 1 2$, the probability of obtaining an even number of 1 is
$
\frac 1 2(1+(1-2\lambda)^t)
$.
Recalling that we have posed $\rho=1-2\lambda$, and replacing $t$ with $|U|$ or $|V|$, we therefore have
$$
U_e \ = \ \frac 1 2 (1+\rho^{|U|})
\qquad\qquad
V_e \ = \ \frac 1 2 (1+\rho^{|V|})
$$
Plugging the above two equations  into \eqref{eq:three}, with some algebra we get
\begin{eqnarray*}
\Pr(Y_\sigma=1,Y_{\sigma'}=1) & = & \frac 1 2 \left(1-\frac 1 2 (1+\rho^{|U|})- \frac 1 2 (1+\rho^{|V|})
 + \frac 1 2 (1+\rho^{|U|})(1+\rho^{|V|})\right)\\[5pt]
                                                                        & = & \frac 1 2 \left(-\frac{\rho^{|U|}}2 -\frac{\rho^{|V|}}2
                                                                        + \frac 1 2\left(1+ \rho^{|U|}+\rho^{|U|}+\rho^{|U|+|V|}\right)\right)\\[5pt]
                                                                        & = & \frac 1 4 (1+\rho^d)\,.
\end{eqnarray*}
This completes  proof that the probability that both $\sigma$ and $\sigma'$ survive  $s=1$   constraint    is $ \frac 1 4 (1+\rho^d)$. The general case of $s\geq 1$ independent constraints  is immediate.
\else
\begin{proof}[Sketch]
The first claim holds by construction. For the second claim, the crucial thing to prove is that, by fixing
just one constraint $c_1$ (so, $s=1$), we have that $\Pr(Y_\sigma=1,Y_{\sigma'}=1) =  \frac {1+\rho^d} 4$.
To this aim, call $A$ and $B$ the sets of variables that are assigned value 1 in $\sigma$ and $\sigma'$, respectively; then, we let   $U=A\setminus B$, $V=B\setminus A$ and $I=A\cap B$.
 Let $C$ be the set of  variables appearing in the constraint $c_1$;
then, $U_e$ and $U_o$ abbreviate $\Pr( |C\cap U|\text{ is even})$ and $\Pr( |C\cap U|\text{ is odd})$, respectively; similarly for $I_e,I_o$ and $V_e,V_o$. Then,
$$
\begin{array}{l}
\Pr(Y_\sigma=1,Y_{\sigma'}=1|a_0=0) = U_e I_e V_e+U_o I_o V_o
\vspace*{.2cm}
\\
\Pr(Y_\sigma=1,Y_{\sigma'}=1|a_0=1) = U_o V_o + U_e V_e-  U_o I_o V_o - U_e I_e V_e
\end{array}
$$
By elementary probability theory, 
$\Pr(Y_\sigma=1,Y_{\sigma'}=1) = \frac 1 2 (1- V_e- U_e+ 2 U_e V_e )$;
the result is obtained by noting that $U_e \ = \ \frac 1 2 (1+\rho^{|U|})$,
$V_e \ = \ \frac 1 2 (1+\rho^{|V|})$ and $d=|U\cup V|=|U|+|V|$.
\fi
\qed
\end{proof}

Now  let
$T_s$ be the random variable  that counts the number of models of $\chi_s$, when the  constraints $c_1,...,c_s$ are chosen independently according to distribution \eqref{eq:distr}:
\begin{eqnarray*}
T_s & \defi & \#\chi_s\,.
\end{eqnarray*}
The event that $\chi_s$ is unsatisfiable   can be expressed as $T_s=0$. A first step toward establishing conditions under which $A$ yields a decision algorithm is the following result. It makes it clear that a possible strategy is to keep  under control the variance of $T_s$,  which depends in turn on $\lambda$. Let us denote by $\mu_s$ the expectation of $T_s$ and by $\var(T_s)$ its variance. Note that $\var(T_s)>0$ if $\#\phi>0$. 

\begin{proposition}\label{prop:variance} \
\begin{enumerate}
\item $\#\phi > 2^{s+\beta}$ implies $\Pr\left(A(s,\phi)=0 \right) \leq \frac{1}{1+\frac {\mu_s^2}{\var(T_s)}}$;
\item $\#\phi < 2^{s-\alpha}$ implies  $\Pr\left(A(s,\phi)=1 \right) < 2^{-\alpha}$.
\end{enumerate}
\end{proposition}
\iffull
\begin{proof}
For part 1, 
\begin{eqnarray*}
  \Pr\left(A(s,\phi)=0 \right) & = & \Pr(T=0)\ \leq\ \frac{\var(T)}{\var(T)+\mu^2}\ =\ \frac{1}{1+\frac{\mu^2}{\var(T)}}
\end{eqnarray*}
where the last but one step is a version of the Cantelli-Chebyshev inequality for integer nonnegative random variables  (see e.g. \cite[Ch.2, Ex.2]{Lugosi}, aka Alon-Spencer's inequality.)

For part 2, first recall from \cite{Gomes} that $\mu =  \frac {\#\phi}{2^s}$
(this can be obtained by Lemma \ref{lemma:basic}(1) and by observing that
$T = \sum_\sigma Y_\sigma$: indeed,
$\mu=E[T] = E[\sum_\sigma Y_\sigma] = \sum_\sigma E[Y_\sigma] = \sum_\sigma 2^{-s} =   \frac {\#\phi}{2^s}$).
To conclude,
\begin{eqnarray*}
  \Pr\left(A(s,\phi)=1 \right) & = & \Pr(T \geq 1)\ \leq\ \mu\ <\ 2^{-\alpha}
\end{eqnarray*}
where the last but one step is Markov's inequality and the last one follows by noting that
$  \frac {\#\phi}{2^s} < 2^{-\alpha}$ under the given hypothesis.
\else
\begin{proof}[Sketch]
The first claim relies on a version of the Cantelli-Chebyshev inequality for integer nonnegative random variables  
(a.k.a. Alon-Spencer's inequality); the second claim relies on Lemma \ref{lemma:basic}(1) and Markov's inequality.
\fi
\qed
\end{proof}

By the previous proposition, assuming $\alpha>1$,
we obtain a decision algorithm (Definition \ref{def:decalg}) provided  that   $\var(T_s)<\mu_s^2$.
  This will depend on the value of $\lambda$ that is chosen, which leads to the following definition.

\begin{definition}[feasibility]\label{def:feas} Let $\phi$, $s$ and $\beta$ be given. A value
$\lambda\in \left(0,\frac 12\right]$ is said to be $(\phi,s,\beta)$-\emph{feasible}   if $\#\phi> 2^{s+\beta}$ implies $\var(T_s)<\mu_s^2$, where
   the constraints in $\chi_s$ are chosen according to  \eqref{eq:distr}.
\end{definition}


Our goal is now to give a method to minimize $\lambda$ while preserving feasibility.
Recall that   $T_s\defi \#\chi_s$.  Denote by $\sigma_1,...,\sigma_N$   the distinct models of $\phi$  (hence, $T_s \leq N$). Note that $T_s=\sum_{i=1}^N Y_{\sigma_i}$. Given any two models $\sigma_i$ and $\sigma_j$, $1\leq i,j\leq N$, let $d_{ij}$ denote their Hamming distance.
The following lemma  gives exact formulae for the expected value and  variance of $T_s$.

\begin{lemma}\label{cor:variance} Let $\rho=1-2\lambda$.
\begin{enumerate}
\item $\mu_s=E[T_s]= N 2^{-s}$;
\item $\var(T_s)= \mu_s + 4^{-s}\sum_{i=1}^N\sum_{j\neq i}    ( {1+\rho^{d_{ij}}}  )^s -
\mu^2_s$;
\item If $N \neq 0$, then $ \frac{\mathrm{var}(T_s)}{\mu^2} = \mu^{-1}_s   + N^{-2}
\sum_{i=1}^N\sum_{j\neq i}    ( {1+\rho^{d_{ij}}}  )^s -1$
\end{enumerate}
\end{lemma}
\iffull
\begin{proof}
The first item is obvious from Lemma \ref{lemma:basic}(1), $T=\sum_{i=1}^N Y_{\sigma_i}$ and linearity of expectation.

Concerning the second item, recall that $\var(T)=E[T^2]-\mu^2$. Now, taking into account that $Y_\sigma\cdot Y_\sigma=Y_\sigma$, we have
\begin{eqnarray*}
T^2 & = & \left(\sum_{i=1}^N Y_{\sigma_i}\right)\left(\sum_{j=1}^N Y_{\sigma_j}\right)
\ =\ \sum_{i=1}^N Y_{\sigma_i}+\sum_{i=1}^N \sum_{j\neq i} Y_{\sigma_i}Y_{\sigma_j}\,.
\end{eqnarray*}
Applying expectation to both sides of the above equation, then exploiting linearity  and    Lemma \ref{lemma:basic}(2)
\begin{eqnarray*}
E[T^2] & = & \frac N{2^s}+\sum_{i=1}^N\sum_{j\neq i}\left(\frac{1+\rho^{d_{ij}}}4\right)^s\,.
\end{eqnarray*}
Since $\mu^2=(\frac N {2^s})^2$, the thesis for this part easily follows.

The third item is an immediate consequence of the first two.
\else
\begin{proof}[Sketch]
The first two items are a direct consequence of Lemma \ref{lemma:basic} and linearity of expectation;
the third item derives from the previous ones.
\fi
\qed
\end{proof}

Looking at the third item above, we clearly see that the upper bound on the error probability  we are after   depends much on `how sparse' the set of $\phi$'s models is in the Hamming space $\mathbb{F}_2^m$: the sparser, the greater the distance, the lower the value of the double summation, the better. Let us denote by $S$ the double summation in the third item of the above lemma:
\begin{eqnarray*}
S &\defi & \sum_{i=1}^N\sum_{j\neq i}    ( {1+\rho^{d_{ij}}} )^s
\end{eqnarray*}
In what follows, we will give an upper bound on $S$ which  is easy to compute and depends on the minimum Hamming distance $d$ among any two models of $\phi$. We need some notation about models of a formula.
\iffull
To this aim, we need a preliminary lemma.

\begin{lemma}\label{lemma:amp}
Let $\Sigma \subseteq \{0,1\}^m$ be such that the minimum Hamming distance between two
strings in $\Sigma$ is $d$. Fix $\sigma \in \Sigma$.  For every $\sigma' \in \Sigma$, let
$\Delta_{\sigma'}=\{i\in\{1,...,m\}:\sigma'(i)\neq \sigma(j)\}$
the set of positions where $\sigma'$ differs from $\sigma$.
For every $j=d,...,m$, let us define the set of integers
$$
L_j^\sigma(\Sigma) \defi \{\ |\Delta_{\sigma'} \cap \Delta_{\sigma''}| \ :\
\{\sigma',\sigma''\} \subseteq \Sigma\ \wedge\ |\Delta_{\sigma'}| = |\Delta_{\sigma''}| = j\ \}
$$
Then
$$
L_j^\sigma(\Sigma) = \left\{
\begin{array}{ll}
\{0,\ldots,j-\left\lceil \frac d 2\right\rceil\} & \mbox{ if } j\leq \frac m 2
\vspace*{.2cm}\\
\{2j-m,\ldots,j-\left\lceil \frac d 2\right\rceil\} & \mbox{otherwise.}
\end{array}
\right.
$$
\end{lemma}
\begin{proof}
Let us first consider how many elements can at most $\Delta_{\sigma'}$ and $\Delta_{\sigma''}$ have in common.
Independently of the value of $j$, since $\sigma'$ and $\sigma''$ must differ for at least $d$ positions,
$\Delta_{\sigma'}$ and $\Delta_{\sigma''}$ can share at most $j-\left\lceil \frac d 2 \right\rceil$ elements.

Let us now consider the least number of elements that $\Delta_{\sigma'}$ and $\Delta_{\sigma''}$
can have in common.
We start by considering $j \leq \frac m 2$ (notice that, if $j = \frac m 2$, then $m$ must be even).
In this case, it is possible that $\Delta_{\sigma'} \cap \Delta_{\sigma''} = \emptyset$,
(e.g., by having all elements of $\Delta_{\sigma'}$ in the first $\left\lfloor \frac m 2 \right\rfloor$ bits
and all elements of $\Delta_{\sigma''}$ in the last $\left\lfloor \frac m 2 \right\rfloor$ bits).
Let $j > \frac m 2$. In this case, $\Delta_{\sigma'}$ and $\Delta_{\sigma''}$ must intersect (since $2j > m$)
and the minimum value for the size of their intersection is $2j-m$.
\qed
\end{proof}
\fi
Below, we let  $j=d,...,m$.
\begin{eqnarray}
l_j &\defi & \left\{\begin{array}{ll}j-\left\lceil \frac d 2\right\rceil +1 & \text{ if }j\leq \frac m 2
\vspace*{.2cm}\\
\max\{0,m-j-\left\lceil \frac d 2\right\rceil +1\} \quad& \text{ if }j> \frac m 2
\end{array}
\right.\label{eq:lj}
\nonumber\\
w^*  & \defi  & \min \left\{w:d\leq w\leq m\text{ and }\sum_{j=d}^w{m\choose {l_j}} \geq N-1 \right\}
\label{eq:wstar}\\
N^* & \defi & \sum_{j=d}^{w^*-1}{m\choose {l_j}}
\label{eq:Nstar}
\end{eqnarray}
where we stipulate that $\min\es=0$.  Note that the definitions of $w^*$ and $N^*$ depend  solely on $N,m$ and $d$.

\iffull
As another ingredient of the proof, we need the following lemma  from extremal set theory; for a proof, see \cite[Th.4.2]{Babai}.

\begin{lemma}[Ray-Chaudhuri-Wilson]\label{lemma:RCW} Let $\mathcal{F}$ be a family of subsets of $\{1,...,m\}$, all of the same size $k$. Let $L$ be a set of nonnegative integers such that $|L|\leq k$ and for every pair of distinct $A,B\in \mathcal{F}$, $|A\cap B|\in L$. Then  $|\mathcal{F}|\leq {m \choose  {|L|}}$.
\end{lemma}
\fi

With the above definitions and results, we have the following upper bound on $S$.

\begin{lemma}\label{lemma:S}
Let the minimal distance between any two models of $\phi$ be at least $d$.
Then
\begin{eqnarray}
\iffull \label{eq:SBound} \fi
S \leq N \left(\sum_{j=d}^{w^*-1} {m\choose {l_j}}  ( {1+\rho^{j}} )^s + (N-1-N^*) ( {1+\rho^{w^*}} )^s\right)
\end{eqnarray}
\end{lemma}
\iffull
\begin{proof}
Fix one of the models of $\phi$ (say $\sigma_i$), and consider the sub-summation originated by it, $S_i\defi \sum_{j\neq i} \left(\frac{1+\rho^{d_{ij}}} 4\right)^s$. Let us group  the remaining $N-1$ models into disjoint  families, $\F_d,\F_{d+1},\ldots$,   of models that are at distance $d,d+1,...$, respectively, from $\sigma_i$.
Note that each of  the $N-1$ models gives rise to exactly  one term  in the summation $S_i$. Hence,
\begin{eqnarray}\label{eq:SiBound}
S_i & = & \sum_{j=d}^{m} |\F_j| \left(\frac{1+\rho^{j}}4\right)^s
\end{eqnarray}
We now want to upper-bound this sum.

To this aim, notice that, for every $\{\sigma',\sigma''\} \subseteq \F_j$. by definition $|\Delta_{\sigma'} \cap \Delta_{\sigma''}| \in L_j^{\sigma_i}(\Sigma)$.
Now, we can identify each model $\sigma$ with $\Delta_\sigma$, its symmetric difference from $\sigma_i$,
seen as a subset of $\{1,...,m\}$. Then, by definition, each set in $\F_j$ has cardinality $j$ and,
by Lemma \ref{lemma:amp} (where $\sigma_i$ plays the role of $\sigma$)
and by Lemma \ref{lemma:RCW}, we have  that $|\F_j|\leq { m \choose {l_j}}$.
Hence, upper-bounding \eqref{eq:SiBound}\footnote{
    Notice that we are not trying to upper bound every summand of \eqref{eq:SiBound}, but just the overall sum.
} consists, e.g., in choosing a tuple of integers $x_d, \ldots, x_m$
in such a way that function
\begin{eqnarray}\label{eq:target}
\sum_{j=d}^{m} x_j \left(\frac{1+\rho^{j}}4\right)^s
\end{eqnarray}
is greater than or equal to the r.h.s. of \eqref{eq:SiBound}, under the constraints
\begin{eqnarray}\label{eq:constraints}
0 \leq x_j \leq {m\choose {l_j}} & \mbox{for } j=d,\ldots,m \ \quad& \sum_{j=d}^{m} x_j = N-1
\end{eqnarray}

Since $0\leq \rho <1$, we obtain an upper-bound by maximizing the $x_j$'s for which $j$ is closer to $d$.
Hence, the optimal solution of the linear programming problem with \eqref{eq:target} as objective function and
\eqref{eq:constraints} as constraints is
\begin{eqnarray*}\label{eq:optimal}
x_j & =&
\left\{
\begin{array}{ll}
{m\choose {l_j}} & \mbox{ for } j=d,\ldots,w^*-1
\vspace*{.2cm}\\
N-1-N^*  & \mbox{ for } j = w^*
\vspace*{.2cm}\\
0 & \mbox{ for } j > w^*
\end{array}
\right.
\end{eqnarray*}
with the definitions of $w^*$ and $N^*$ given in \eqref{eq:wstar} and  \eqref{eq:Nstar}, respectively.


All together, these facts imply that
\begin{eqnarray*}
S_i & \leq & \sum_{j=d}^{w^*-1} {m\choose {l_j}} \left(\frac{1+\rho^{j}}4\right)^s + (N-1-N^*)\left(\frac{1+\rho^{w^*}}4\right)^s
\end{eqnarray*}
By summing the above inequality over all models $\sigma_i$, we obtain the thesis.
\else
\begin{proof}[Sketch]
Fix one of the models of $\phi$ (say $\sigma_i$), and consider the sub-summation originated by it, $S_i\defi \sum_{j\neq i} \left(\frac{1+\rho^{d_{ij}}} 4\right)^s$. Let us group  the remaining $N-1$ models into disjoint  
families, $\F_d,\F_{d+1},\ldots$,   of models that are at distance $d,d+1,...$, respectively, from $\sigma_i$.
Note that each of  the $N-1$ models gives rise to exactly  one term  in the summation $S_i$. Hence,
$S_i = \sum_{j=d}^{m} |\F_j| \left(\frac{1+\rho^{j}}4\right)^s$.
By the Ray-Chaudhuri-Wilson Lemma \cite[Th.4.2]{Babai}, $|\F_j|\leq { m \choose {l_j}}$.
Hence, upper-bounding $S_i$ consists, e.g., in choosing a tuple of integers $x_d, \ldots, x_m$ such that
$\sum_{j=d}^{m} x_j \left(\frac{1+\rho^{j}}4\right)^s \geq \sum_{j=d}^{m} |\F_j| \left(\frac{1+\rho^{j}}4\right)^s$, 
under the constraints $0 \leq x_j \leq {m\choose {l_j}}$, for $j=d,\ldots,m$, and $\sum_{j=d}^{m} x_j = N-1$.
An optimal solution is
\begin{eqnarray*}\label{eq:optimal}
x_j & =&
\left\{
\begin{array}{ll}
{m\choose {l_j}} & \mbox{ for } j=d,\ldots,w^*-1
\vspace*{.2cm}\\
N-1-N^*  & \mbox{ for } j = w^*
\vspace*{.2cm}\\
0 & \mbox{ for } j > w^*
\end{array}
\right.
\end{eqnarray*}
The thesis is obtained by summing over all models $\sigma_i$.
\fi
\qed
\end{proof}

\begin{definition}\label{def:bound}
Given $s\geq 1$, $\beta>0$, $d\geq 1$ and $\lambda\in \left(0,\frac12\right]$, let us define
{
\begin{equation} \nonumber
\begin{array}{lrr}
B(s,m,\beta,d,\lambda)\defi & \multicolumn{2}{r}{2^{-\beta} \!\!+\!2^{-s-\beta}\!\left(\sum_{j=d}^{w^*-1}\!\! {m\choose {l_j}}  ( {1+\rho^{j}} )^s  \!+\!  (N \!-\! 1 \!-\! N^*)\! ( {1+\rho^{w^*}} )^s\!\right) \!-\! 1,}
\end{array}
\end{equation}
}
\!\!where $\rho \defi 1 - 2\lambda$ and $N=\lceil 2^{s+\beta}\rceil $ also in the definition of $w^*$ and $N^*$.\end{definition}

Using the facts collected so far,  the following theorem follows,
giving an upper bound on $\frac {\var(T_s)}{\mu^2_s}$.

\begin{theorem}[upper bound]\label{th:mainBound}
Let   the minimal distance between  models of $\phi$ be at least $ d$ and $\#\phi > 2^{s+\beta}$. Then,
$\frac{\var(T_s)}{\mu^2_s} \leq   B(s,m,\beta,d,\lambda)$.
\end{theorem}
\begin{proof}
First note that we can assume without loss of generality that $\#\phi=N=\lceil 2^{s+\beta}\rceil$.
If this was not the case, we can consider in what follows
any formula $\phi'$ whose models are models of $\phi$
but are exactly $\lceil 2^{s+\beta}\rceil$ ($\phi'$ can be obtained by adding some conjuncts to $\phi$
that exclude $\#\phi - \lceil 2^{s+\beta}\rceil$ models). Then, $\Pr\left(A(s,\phi)=0 \right) \leq
\Pr\left(A(s,\phi')=0 \right)$ and this would suffice, for the purpose of upper-bounding
$\Pr\left(A(s,\phi)=0 \right) $.
\iffull

Then, by Lemma \ref{cor:variance}(3) (first step), by Lemma \ref{lemma:S} (second step)
and by the fact that $N = \lceil 2^{s+\beta}\rceil \geq 2^{s+\beta}$ (fourth step), we have :
\begin{eqnarray*}
\frac{var(T)}{\mu^2} & = &
\frac{2^s}{N} + \left(\frac{2^s}{N}\right)^2 S -1
\\
& \leq &
\frac{2^s}{N} + \left(\frac{2^s}{N}\right)^2 N \left(\sum_{j=d}^{w^*-1} {m\choose {l_j}} \left(\frac{1+\rho^{j}}4\right)^s \right.\\
&& \hspace*{1.7cm} \left. +\ (N-1-N^*)\left(\frac{1+\rho^{w^*}}4\right)^s\right) - 1
\\
& = &
\frac{2^s}{N} + \frac{2^{2s}}{N} \left(\sum_{j=d}^{w^*-1} {m\choose {l_j}} \left(\frac{1+\rho^{j}}4\right)^s
 +\  (N-1-N^*)\left(\frac{1+\rho^{w^*}}4\right)^s\right) - 1
\\
& \leq &
\frac{2^s}{2^{s+\beta}} + \frac{2^{2s}}{2^{s+\beta}} \left(\sum_{j=d}^{w^*-1} {m\choose {l_j}} \left(\frac{1+\rho^{j}}4\right)^s
\right.\\
&& \hspace*{1cm} \left. +\  (\lceil 2^{s+\beta}\rceil-1-N^*)\left(\frac{1+\rho^{w^*}}4\right)^s\right) - 1
\\
& = &
B(s,m,\beta,d,\lambda)
\end{eqnarray*}
The result follows from Proposition \ref{prop:variance}(1).
\fi
\qed
\end{proof}

The following notation will be useful in the rest of the paper. For
$0\leq \gamma\leq 1$, define
\begin{eqnarray}\label{eq:lambdaStar}
\lambda^*_\gamma(s,m,\beta,d)  \defi  \inf \left\{\lambda\in  \left(0,
 \frac 1 2  \right]\,:\,B(s,m,\beta,d,\lambda)\leq \gamma \right\}
\end{eqnarray}
where we stipulate $\inf \es =+\infty$.

\begin{corollary}[Feasibility]\label{cor:feas}
Assume  the minimal distance between any two models of $\phi$ is at least $d$.
 Then every   $\lambda\in \left(\lambda^*_1(s,m,\beta,d) , \frac 1 2 \right]$ is $(\phi,s,\beta)$-feasible.
\end{corollary}
\iffull
\begin{proof}
Since $B(s,m,\beta,d,\lambda)$ is a decreasing function of $\lambda$, whenever it exists the value $\lambda$ s.t. $B(s,m,\beta,d,\lambda)=1$ it is unique. If $\lambda^* \neq +\infty$,
since $B(s,m,\beta,d,\lambda)$ is an upper
bound of $\frac {\var(T)}{\mu^2}$, any value $\lambda$ such that $\frac 12\geq \lambda >\lambda^*(s,m,\beta,d)$ is feasible and therefore can be used to make algorithm $A$ work (see Theorem \ref{th:mainBound}).
If $\lambda^*=+\infty$, we note that $\lambda=\frac 12$ is feasible:
\begin{eqnarray*}
\frac{\var(T)}{\mu^2} \ \leq \ \frac{2^s}{N} + \frac{2^{2s}}{N^2} \cdot N \cdot \frac{N^* + N - 1 - N^*}{4^s} -1
\ < \ \frac{2^s}{N} \ \leq \ 1
\end{eqnarray*}
where the first step comes from Lemma \ref{cor:variance}(3) and \eqref{eq:SBound} (with $\lambda = \frac 12$) and the third one holds by hypothesis.
\qed
\end{proof}
\fi

\section{Evaluation}\label{sec:four}

\subsection{Theoretical bounds on expected constraint length}

To assess the improvements that our theory introduces on XOR-based
approximate counting, we start by considering the impact of minimum
Hamming distance on the expected length of the XOR constraints.
First, in Fig. \ref{fig:plot} we plot  $\lambda^*_1$ as a
function of $s$, for fixed values of $m=32$ and $64$, $\beta=1.5$,
and four different values of $d$. Note that the difference between
different values of $d$ tends to vanish as $s$ gets large -- i.e.
close to $m$.

\begin{figure}[t]
\begin{center}
        \includegraphics[width=0.49\textwidth]{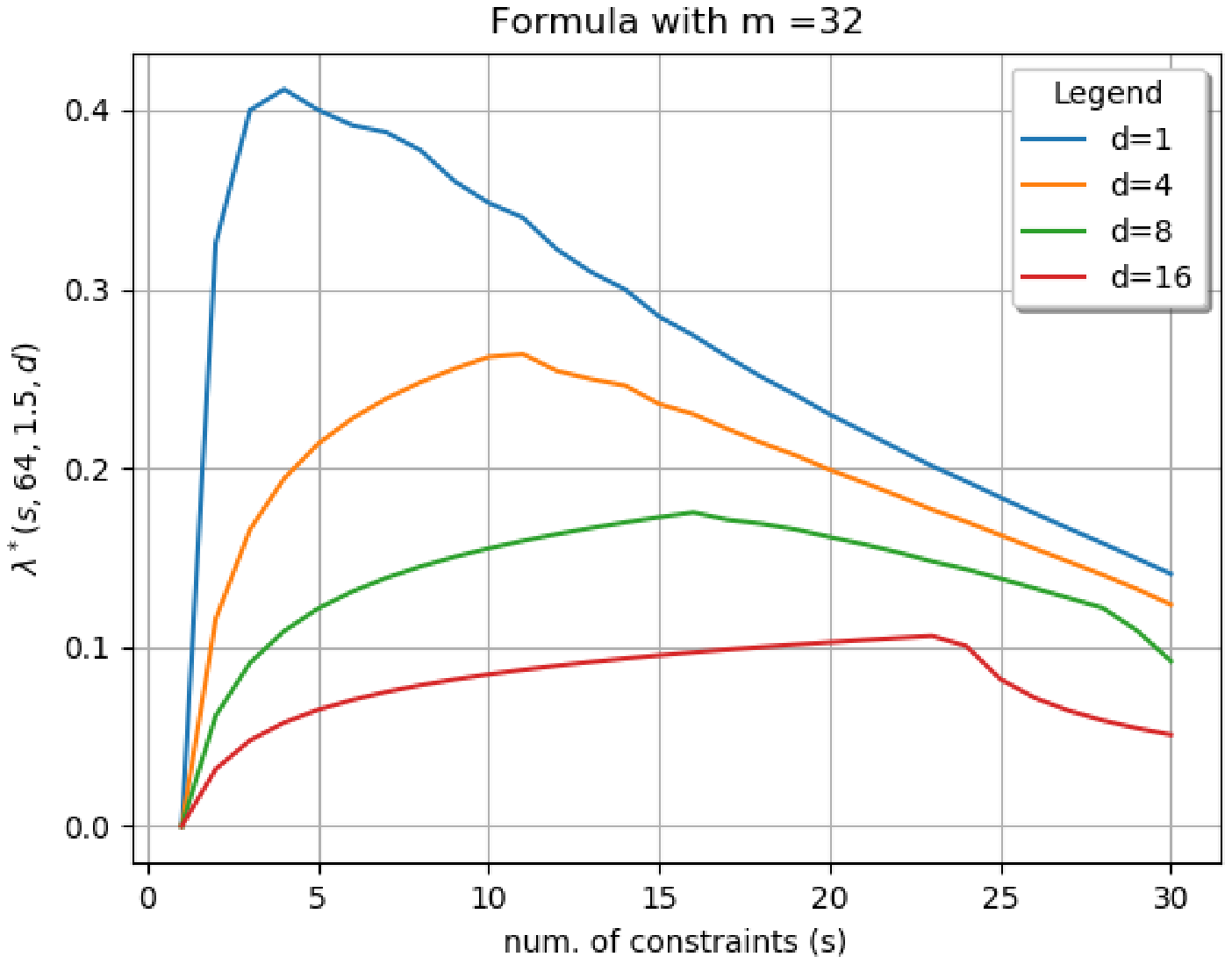}
        \includegraphics[width=0.49\textwidth]{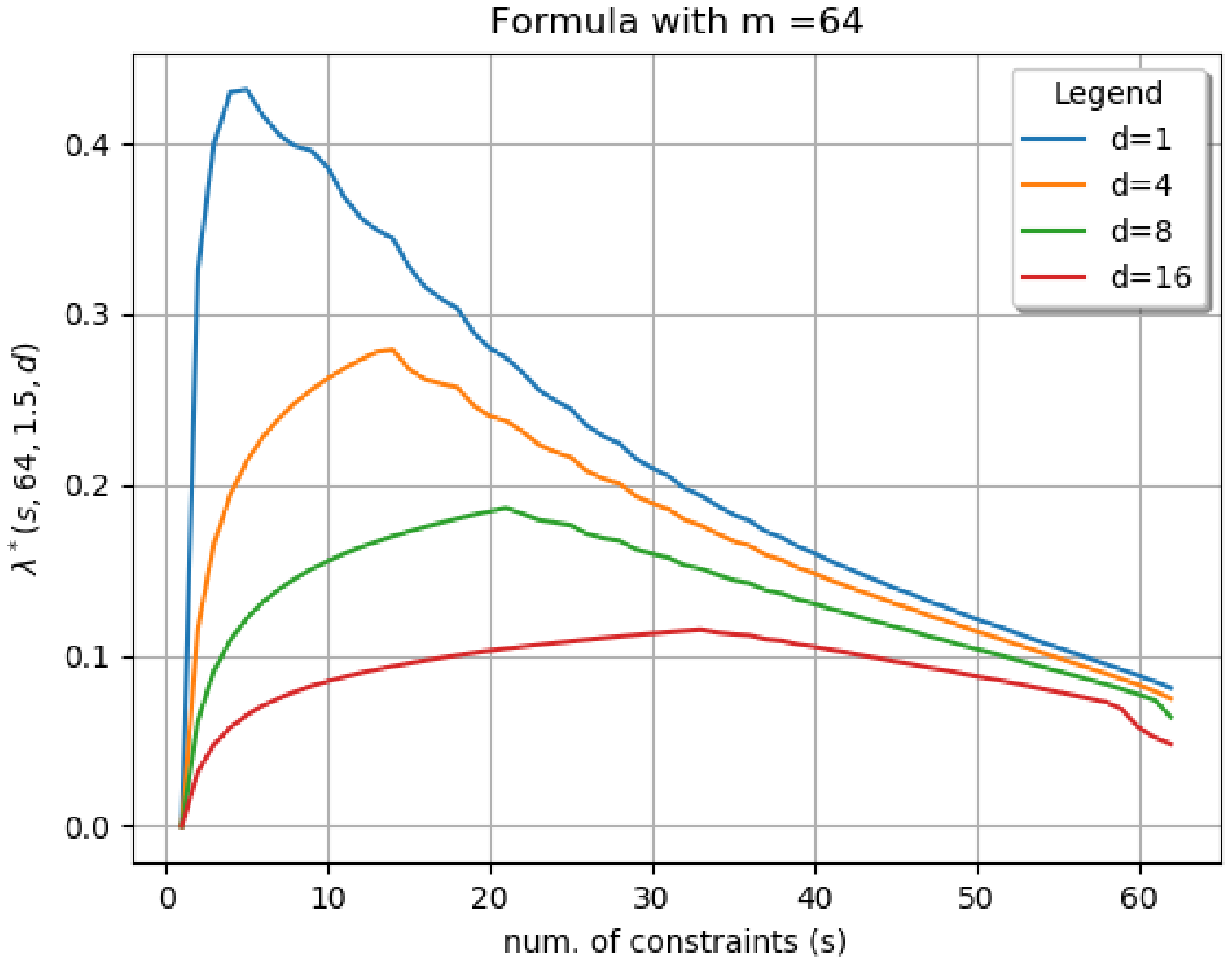}
        \vspace*{-.4cm}
\end{center}
        \caption{Plots of $\lambda^*_1$ as a
        function of $s$, for  $m=32$ and $m=64$, $\beta=1.5$,
        and different values of $d$. For any value of $s$ and $d$, any value of $\lambda$
        above the curve is feasible.} \label{fig:plot}
\end{figure}

Next, we compare our theoretical bounds with those in  \cite{EGSS14},
where a goal similar to ours is pursued. Interestingly,
their  bounds  coincide with ours when setting
$d=1$ --  no assumption on the minimum Hamming distance --
  showing that our approach generalizes theirs.
We report a numerical comparison in Table
\ref{tab:compare}, where   several hypothetical values of
$m$ (no. of variables) and $s$ (no. of constraints) 
 are considered.
 Following a similar evaluation conducted  in \cite[Tab.1]{EGSS14}, here we fix
 the error probability to $\delta=\frac 4 9$ and  the upper slack parameter to $\beta=2$,
  and report the values of
 $\lambda \times m$ for the minimal value of $\lambda$
 that would guarantee  a confidence of at least $1- \delta$ in case an upper bound is found,  computed with their approach and ours.
Specifically, in their case
   $\lambda$ is obtained via the formulae in   \cite[Cor.1,Th.3]{EGSS14}, while in our case
     $\lambda = \lambda^*_{\gamma}$ for $\gamma=0.8$, which entails the
   wanted confidence according to Proposition \ref{prop:variance}(2).
We see that, under the assumption that     lower bounds on $d$  as illustrated  are known,
in some cases a dramatic reduction of the expected length of the XOR constraints  is obtained.

\begin{table*}[t]
\begin{center}
\begin{tabular}{|rr|r|r|r|r|r|}
\hline
\quad {\bf\em N. Vars}  \qquad & {\bf\em N. Constraints} & {\quad   $ {\lambda \times m}$    } & \multicolumn{3}{c}{
   $ {\lambda \times m}$ {\bf\em present paper}} \vline\\
{\bf\em (m)} & {\bf\em (s)} &  \quad {\bf\em from \cite{EGSS14}} & {\quad   $d=1$}  &  {\quad   $d=5$}  & {   $d=20$}\\
\hline
 50  & 13 & 16.85 & 16.85 & 11.76 & 3.88\\
 50  & 16 & 15.38 & 15.38 & 11.36 & 4.1\\
 50  & 20 & 13.26 & 13.26 & 10.26 & 4.37\\
 50  & 30 & 9.57 & 9.57 & 8.02 & 4.75\\
 50  & 39 & 7.08 & 7.08 & 6.2 & 4.45\\
 100  & 11 & 39.05 & 39.05 & 23.65 & 7.4\\
 100  & 15 & 35.44 & 35.44 & 25.26 & 8.07\\
 100  & 25 & 27.09 & 27.09 & 21.33 & 9.14\\
 119  & 7 & 50.19 & 50.18 & 25.07 & 7.6\\
 136  & 9 & 55.63 & 55.63 & 30.66 & 9.46\\
 149  & 11 & 60.6 & 60.6 & 35.24 & 11.02\\
 352  & 10 & 147.99 & 147.99 & 81.42 & 25.31\\
\hline
\end{tabular}
\end{center}
\caption{Comparison with the provable bounds   from \cite[Tab. 1]{EGSS14}.}
\vspace*{-.4cm}
\label{tab:compare}
\end{table*}

\subsection{Execution times}
Although the focus of the present paper is mostly theoretical, it is
instructive to look at the results of some simple experiments for a
first concrete assessment of the proposed methodology. To this aim, we
have implemented in Python the algorithm $C_A$ with $A$ as described
in Section \ref{sec:three},  relying on   CryptoMiniSAT  \cite{CMS}
as a SAT solver, and conducted  a few experiments\footnote{Run on a
MacBook Air, with a 1,7 GHz Intel Core i7 processor, 8 GB of memory
(1600 MHz DDR3) and OS X 10.9.5.}.

The crucial issue to use Theorem \ref{th:mainBound} is the knowledge
of  (a lower bound on) the minimal distance $d$ among the models of
the formula we are inputting to our algorithm. In general, this
information is unknown and we shall discuss a possible approach to
face this problem in the next section. For the moment, we use a few
 formulae describing the set of codewords of certain error correcting codes, for which the number of models and the
minimum distance   is known by construction. \iffull

More precisely, we consider BCH codes \cite{BCH1,BCH2},  that are
very well-known error-correcting codes whose minimal distance among
the codewords is lower-bounded by construction. The property of such
codes we shall use is the following:
\begin{quote}
For every $q \geq 3$ and $t < 2^q-1$, there exists a BCH code such
that codewords have $n = 2^q-1$ bits, including $k \geq n-qt$
parity-check bits, and their minimum distance is $d_{min} \geq 2t +
1$.
\end{quote}
From the theory of error-correcting codes, we know that every binary
linear $[n,m,d]$ block code $\Code$ (i.e. a code with codeword
length $n$, message length $m$ and minimum distance $d$)  can be
defined by its parity check matrix $H$, i.e. a $m \times n$ (with
$m\defi n-k$) binary matrix such that $u\in \mathbb{F}_2^n$ is a
codeword of $\Code$ if and only if
\begin{eqnarray}\label{eq:checkmatr}
Hu^T= \underline 0
\end{eqnarray}
where we use the vector-matrix multiplication in the field
$\mathbb{F}_2$ and $\underline 0$ denotes the $m$-ary zero vector.
Let $\varphi_H(u)$ be the propositional formula encoding that $u$ is
a codeword of $\Code$:
starting from  \eqref{eq:checkmatr},  $\varphi_H(u)$ can be
expressed as a conjunction of   $m$   XOR constraints    (assuming
$H$ has full rank):
$$
c_i \defi \left( \bigoplus_{j=1}^n H_{ij} \cdot u_j \ \ =\ 0 \right
) \quad \mbox{ for } i = 1,\ldots,m
$$
Thus, once fixed $n$, $m$ and $t$, we can calculate the
corresponding matrix $H$ and the associated formula $\varphi_H$; the
latter will have $2^m$ models (every sequence in $\mathbb{F}_2^m$
will yield a different model for $\varphi_H$), with Hamming distance
$d_{min}$.

In this way, we have created the following CNF formulae:

\begin{center}
\begin{tabular}{c|ccc|cc}
{\bf\em Name of the formula} & {\bf\em n} &  {\bf\em m} &  {\bf\em t} &  {\bf\em d}$_{min}$ & {\bf\em num vars}\\
\hline F16-31-7.cnf & 31 & 16 & 3 & 7 & 1221
\\
F21-31-5.cnf & 31 & 21 & 2 & 5 & 1591
\\
F26-31-5.cnf & 31 & 26 & 1 & 3 & 1091
\\
F24-63-15.cnf & 63 & 24 & 7 & 15 & 4862
\\
\end{tabular}
\end{center}

\noindent \fi

For these formulae, we run 3 kinds of experiments.
First, we run the tool for every formula without using the improvements of the bounds
given by knowing the minimum distance (i.e., we used Theorem  \ref{th:mainBound} by setting $d=1$).
Second, we run the tool for every formula by using the known minimum distance.
Third, we run the state-of-the-art tool for approximate model counting,
called ApproxMC3 \cite{SM19} (an improved version of ApproxMC2
\cite{CMV16}). The results obtained are reported in Table
\ref{table:results}.

\begin{table*}[t]
\begin{center}
\begin{tabular}{|c|lll|lll|}
\hline
{\bf\em Formula} & {\bf\em Our tool with d=1} \qquad
&  {\bf\em Our tool with d=d}$_{min}$  &  {\bf\em ApproxMC3}\\
\hline
F21-31-5.cnf & res: $[2^{19.5},2^{23.5}]$ & res: $[2^{19.5},2^{23.5}]$ & res: ??\\
& time: 8.05 secs & time: 6.73 secs & time: $>3$ hours\\
\hline
F16-31-7.cnf & res: $[2^{14.5},2^{18.5}]$ & res: $[2^{14.5},2^{18.5}]$ & res: $[2^{13.86},2^{17.85}]$\\
& time: 11.75 secs & time: 9.5 secs & time: 22 mins 43 secs\\  
\hline
F11-31-9.cnf &  res: $[2^{9.5},2^{13.5}]$ & res: $[2^{9.5},2^{13.5}]$ & res: $[2^{9.09},2^{13.09}]$\\
& time: 6.75 secs & time: 4.32 secs & time:  15.24 secs\\
\hline
F6-31-11.cnf & res: $[2^{4.5},2^{8.5}]$ & res: $[2^{4.5},2^{8.5}]$ & res: $[2^4,2^8]$\\
& time: 3.01 secs & time: 2.62 secs & time: 1.9 secs\\  
\hline
F16-63-23.cnf &  res: $[2^{14.5},2^{18.5}]$ & res: $[2^{14.5},2^{18.5}]$ & res: $[2^{13.9},2^{17.9}]$
\\
& time: 31 mins 15 secs \qquad \quad& time: 2 min 36 secs \qquad \qquad& time: 54 mins 57 secs
\\
\hline
\end{tabular}
\end{center}
\caption{Results for our tool with $\alpha=\beta=1.5$, $d=1$ and $d=d_{min}$, compared to ApproxMC3 with a tolerance $\epsilon = 3$. 
In all trials, the error probability $\delta$ is 0.1.
}
\vspace*{-.4cm}
\label{table:results}
\end{table*}

To compare our results with theirs, we have to consider that, if our tool returns $[l,u]$,
then the number of models lies in $[\lfloor 2^l\rfloor, \lceil 2^u \rceil]$ with error probability $\delta$ (set to 0.1. in all
experiments). By contrast, if \mbox{ApproxMC3}
returns a value $M$, then the number of models lies in
$\left[\frac M {1+\epsilon} , M(1+\epsilon)\right]$ with error
probability $\delta$ (again, set to 0.1 in all experiments).
So, we have to choose for ApproxMC3 a tolerance $\epsilon$ that produces an interval of possible solutions
comparable to what we obtain with our tool.
The ratio between the sup and the inf of our intervals is $2^{2\max(\alpha,\beta)+1}$
(indeed, $A$ always returned an interval such that $u-l \leq 2\max(\alpha,\beta)+1$);
when $\alpha=\beta=1.5$, the value is 16.
By contrast, the ratio between the sup and the inf of ApproxMC3's intervals is $(1+\epsilon)^2$;
this value is 16 for $\epsilon = 3$.

In all formulae that have ``sufficiently many" models --
empirically, at least $2^{11}$ -- our approach outperforms
ApproxMC3. Moreover, making use of the minimum distance information
implies a gain in performance. Of course, the larger the
distance, the greater the gain -- again, provided that the formula
has sufficiently many
models: compare, e.g., the first and the last formula.

\subsection{Towards a  practical methodology}
\label{sec:impr}
%
To be used in practice, our technique requires a lower bound on the
minimum Hamming distance between any two models of $\phi$.
We discuss below how error-correcting codes might be used, in principle, to obtain such a bound. Generally, speaking an error-correcting code
 adds redundancy to a string of bits and
inflates the minimum distance  between the resulting codewords.
The idea here is to transform the given formula $\phi$ 
into a new formula $\phi'$ that describes an encoding of the original formula's models: as a result
  $\#\phi'=\#\phi$,  but the models of $\phi'$ live in a higher dimensional space,
where a minimum distance between models is ensured.

Assume that $\phi(y)$ is already in Conjunctive Normal Form.
Fix a binary linear $[n,m,d]$ block code $\Code$
(i.e. a code with $2^m$ codewords of $n$ bits and minimum distance $d$),
and let $G$ be its generator matrix, i.e. a
$m \times n$ binary matrix such that the codeword associated to $u\in \mathbb{F}_2^m$ is $uG$
(where we use the vector-matrix multiplication in the field $\mathbb{F}_2$).
The fact that a $c\in \mathbb{F}_2^n$ is a codeword can be expressed by finding some $u$
that satisfies the conjunction of the $n$ XOR constraints:
$$
c_i =  \bigoplus_{j=1}^m u_j \cdot G_{ij}  \quad \mbox{ for } i = 1,\ldots,n\,.
$$
Again, the important condition here is that $G$   be \emph{sparse} (on its columns), so that the above formula effectively corresponds to a conjunction of sparse XOR constraints. That is, we should confine ourselves to   low-density parity check (LDPC) codes \cite{Gall62,LDPC}.
Now, we consider the formula
\begin{eqnarray*}
\phi'(z) & \defi & \exists y(\phi(y)\wedge z = yG)\,.
\end{eqnarray*}
\iffull
and prove the following result.

\begin{proposition}
$\#\phi=\#\phi'$.
\end{proposition}
\begin{proof}
Trivially, every model $\sigma$ of $\phi$ induces a model $\sigma G$ for $\phi'$.
For the converse, let $\sigma'$ be a model of $\phi'$; this means that there exists
a $\sigma$ such that $\sigma' = \sigma G$ and $\sigma$ satisfies $\phi$.
Moreover, as we now prove, there is just one such $\sigma$ (so, the function associating
$\sigma$ to $\sigma G$ is bijective). Indeed,
the mapping from $\mathbb{F}_2^m$ to $\mathbb{F}_2^n$ induced by multiplication by $G$ is injective.
To see this, consider $G$ in standard form, i.e. $G = [I_m\ \ A]$, where $I_m$ is the identity matrix of order $m$
(it is well-known that every generator matrix can be put in this form by generating an equivalent code --
see, e.g., \cite[Section 1.2]{Pless}). Thus, assume that $c \defi \sigma G = \sigma'G \defi c'$; then,
$$
c_i = \left\{
\begin{array}{ll}
\sigma_i & \mbox{ for } i = 1,\ldots,m
\vspace*{.3cm}\\
\bigoplus_{j=1}^m \sigma_j \cdot G_{ij} & \mbox{ for } i = m+1,\ldots,n
\end{array}
\right.
$$
For $c'_i$ similar values can be obtained, with $\sigma'$ in place of $\sigma$. Thus, trivially,
$\sigma = \sigma'$.
\qed
\end{proof}

\fi
If we now assume a minimum distance of $d$ when applying  Theorem \ref{th:mainBound}, we have a   decrease in the feasibility threshold $\lambda^*$, as prescribed by \eqref{eq:lambdaStar}. This gain   must of course be balanced against the increased number of boolean variables in the formula (viz. $n$). We will have an actual advantage using $\Code$ over not using it (and simply assuming $d=1$) if and only if, by using $\Code$, the expected length of the resulting XOR constraints is actually smaller. By letting $\lambda^{*,{d}}\defi\lambda^*_1(s,n,\beta,d)$,
 the latter fact holds if and only if
\begin{eqnarray}\label{eq:code}
n\lambda^{*,{d}} & \leq & m\lambda^{*,1}
\end{eqnarray}
or equivalently $\lambda^{*,{d}} \leq R\lambda^{*,{1}}$,
where $R\defi \frac m n$ is the \emph{rate} of $\Code$.
This points to  codes with high rate and big minimum distance.
Despite these two parameters  pull one against the other,  \eqref{eq:code} can be fulfilled, and good expected length bounds obtained,  by choosing $\Code$ appropriately.

\begin{figure}[t]
\begin{center}
        \includegraphics[width=0.49\textwidth]{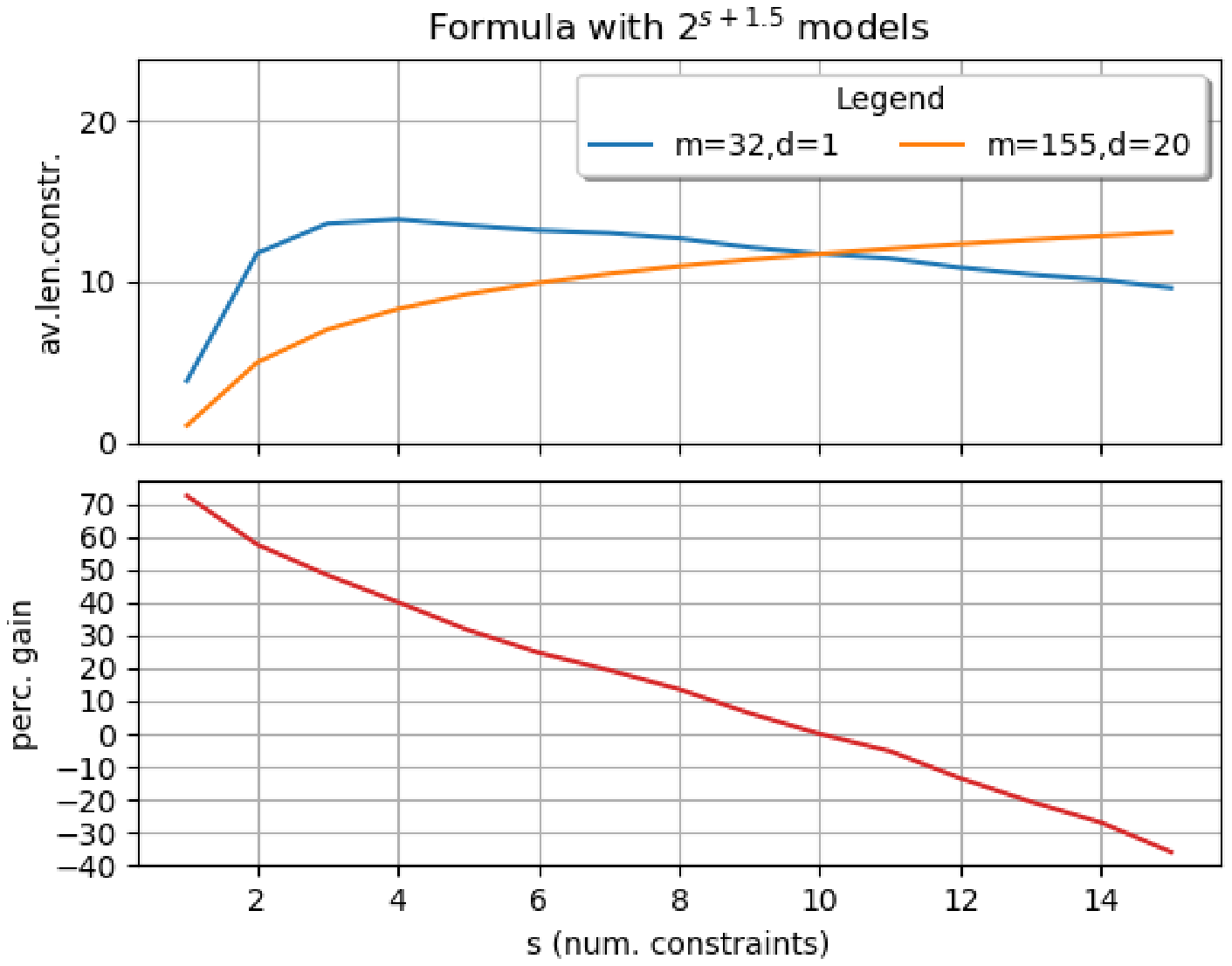}
        \includegraphics[width=0.49\textwidth]{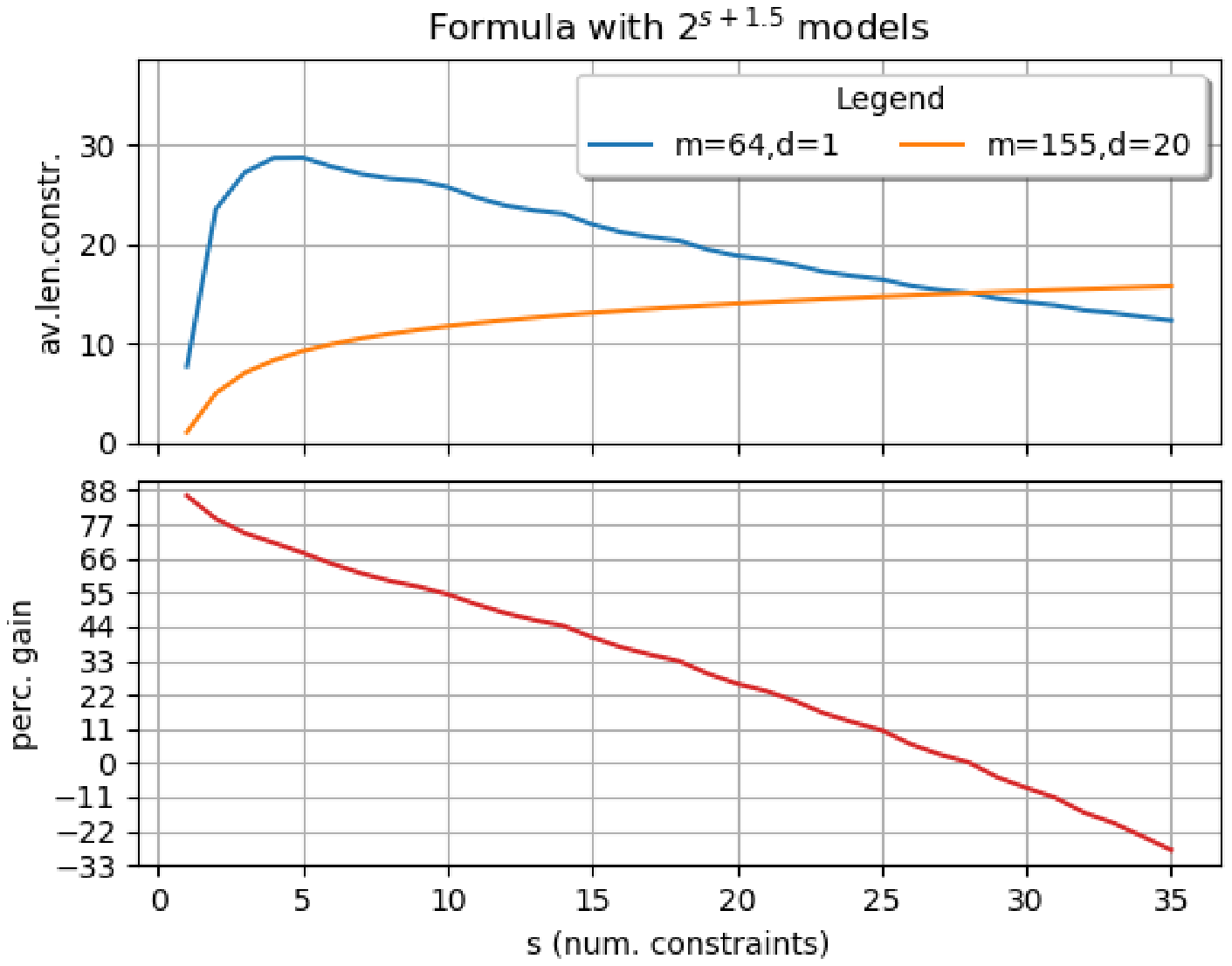}
                \vspace*{-.4cm}
\end{center}
        \caption{Plots of the expected length of XOR constraints as a function of $s$, with and without code, and the relative percentual gain. Here, $m=32$ (left) and $64$ (right), $\beta=1.5$ and the code is a $[155,64,20]$-LDPC.}\label{fig:code}
\end{figure}

For example, \cite{Tanner} presents a $[155,64,20]$-LDPC code, that is a code with block length of 155, with $2^{64}$ codewords and with minimum distance between codewords of 20.
In Fig. \ref{fig:code} we compare the expected length of the resulting XOR constraints in the two cases  -- $m\times \lambda^*(s,m,\beta,1)$ (without the code, for $m=32$ and $m=64$) and $155\times \lambda^*_1(s,155,\beta,20)$ (with the code) -- as functions of $s$, for fixed $\beta=1.5$.
As  seen from the plots, the use of the code offers a significant advantage in terms of expected length
 up to $s=10$ and $s=26$, respectively.

We have performed a few tests for a preliminary practical assessment of this idea. Unfortunately, in all but a few cases,
the use of the said $[155,64,20]$-LDPC code does not offer a significant advantage in terms of execution time.
Indeed, embedding a formula
in this code   implies adding 155 new XOR constraints: the   presence of
so many constraints,  however short,
apparently outweights the   benefit of a   minimum distance $d=20$.
We hope that alternative codes, with
  a more advantageous block length versus minimum distance tradeoff,
  would fare better. 
Indeed, as we showed in Table \ref{tab:compare},     relatively small distances
(e.g. $d=5$) can already give interesting gains, if the number of extra constraints is small.
 We leave that as subject for future research.

\section{Conclusion}\label{sec:five}
We have studied the relation between   sparse XOR constraints and
minimum Hamming distance in model counting. Our findings suggest
that minimum distance plays an important role in making the
feasibility threshold for $\lambda$ (density) lower, thus
potentially improving the effectiveness of XOR based model counting
procedures. These results also prompt a natural direction for future
research:  embedding the set of models into a higher dimensional
Hamming space, so as to
enforce a given minimum distance. 

Beside the already
mentioned \cite{EGSS14}, our work also relates to the recent
\cite{AT17}. There,   constraints    are represented as systems $Ay
= b$, for    $A$ a random LDPC    matrix   enjoying certain
properties,
  $b$ a random vector, 
and $y$   the variable  vector. 
Their results  are  quite different from ours, but also they take the
geometry of the set of models  into account, including minimum
distance. In particular, they   make their bounds depend also on a
 ``{boost}''  parameter which appears quite difficult to compute.  This
leads to a methodology that is only empirically validated -- that
is, the
model count results are offered with no guarantee.  

\paragraph{Acknowledgements }
We   thank 
Marco Baldi, Massimo Battaglioni and Franco Chiaraluce for
providing us with the generator and parity check matrices of the LDPC code in Subsection
\ref{sec:impr}.

\begin{small}
\bibliographystyle{splncs04}
\bibliography{modelCount}
\end{small}


\end{document}